\documentclass[aps,prx,amsmath,twocolumn,amssymb,notitlepage,superscriptaddress]{revtex4-1}
\usepackage{balance}
\usepackage{amssymb}
\usepackage{amsthm}
\usepackage{amsfonts}
\usepackage{complexity}
\usepackage{graphicx}
\usepackage{dcolumn}
\usepackage{bm}
\usepackage{hyperref}
\usepackage[capitalise,nameinlink]{cleveref}
\usepackage{enumerate}
\usepackage{algorithm}
\usepackage{algpseudocode}
\usepackage{float}

\usepackage{multirow}
\usepackage{epstopdf}
\usepackage{tabularx}
\usepackage{ragged2e}
\newtheorem{theorem}{Theorem}

\usepackage[usenames,dvipsnames]{xcolor}
\hypersetup{
    colorlinks=true,       
    linkcolor=Maroon,          
    citecolor=OliveGreen,        
    filecolor=magenta,      
    urlcolor=Blue           
}

\begin{document}

\def\ket#1{\left|#1\right\rangle}
\def\bra#1{\langle#1|}
\newcommand{\ketbra}[2]{|#1\rangle\!\langle#2|}
\newcommand{\braket}[2]{\langle#1|#2\rangle}
\newcommand{\prob}[1]{{\rm Pr}\left(#1 \right)}
\newcommand{\expect}[2]{{\mathbb{E}_{#2}}\!\left\{#1 \right\}}
\newcommand{\var}[2]{{\mathbb{V}_{#2}}\!\left\{#1 \right\}}
\newcommand{\Ad}{\operatorname{Ad}}
\newcommand{\ad}{\operatorname{ad}}
\newcommand{\plog}[1]{\operatorname{polylog}{#1}}

\newcommand{\sinc}{\operatorname{sinc}}
\newcommand{\ghl}[1]{{\color{red}GHL {#1}}}

\newcommand{\sde}{\mathrm{sde}}
\newcommand{\Z}{\mathbb{Z}}
\newcommand{\RR}{\mathbb{R}}
\newcommand{\w}{\omega}
\newcommand{\Kap}{\kappa}

\newcommand{\Tchar}{$T$}
\newcommand{\T}{\Tchar~}
\newcommand{\TT}{\mathrm{T}}
\newcommand{\ClT}{\{{\rm Clifford}, \Tchar\}~}
\newcommand{\Tcount}{\Tchar--count~}
\newcommand{\Tcountper}{\Tchar--count}
\newcommand{\Tcounts}{\Tchar--counts~}
\newcommand{\Tdepth}{\Tchar--depth~}
\newcommand{\Zr}{\Z[i,1/\sqrt{2}]}
\newcommand{\ve}{\varepsilon}
\newcommand{\expectation}[1]{\mathop{\mathbb{E}}\left[#1\right]}
\newcommand{\ssection}[1]{{\textbf{#1 --}}}

\newcommand{\eq}[1]{\hyperref[eq:#1]{(\ref*{eq:#1})}}
\renewcommand{\sec}[1]{\hyperref[sec:#1]{Section~\ref*{sec:#1}}}
\newcommand{\app}[1]{\hyperref[app:#1]{Appendix~\ref*{app:#1}}}
\newcommand{\fig}[1]{\hyperref[fig:#1]{Figure~\ref*{fig:#1}}}
\newcommand{\thm}[1]{\hyperref[thm:#1]{Theorem~\ref*{thm:#1}}}
\newcommand{\lem}[1]{\hyperref[lem:#1]{Lemma~\ref*{lem:#1}}}
\newcommand{\tab}[1]{\hyperref[tab:#1]{Table~\ref*{tab:#1}}}
\newcommand{\cor}[1]{\hyperref[cor:#1]{Corollary~\ref*{cor:#1}}}
\newcommand{\alg}[1]{\hyperref[alg:#1]{Algorithm~\ref*{alg:#1}}}
\newcommand{\defn}[1]{\hyperref[def:#1]{Definition~\ref*{def:#1}}}

\newcommand{\targfix}{\qw {\xy {<0em,0em> \ar @{ - } +<.39em,0em>
\ar @{ - } -<.39em,0em> \ar @{ - } +
<0em,.39em> \ar @{ - }
-<0em,.39em>},<0em,0em>*{\rule{.01em}{.01em}}*+<.8em>\frm{o}
\endxy}}

\newenvironment{proofof}[1]{\begin{trivlist}\item[]{\flushleft\it
Proof of~#1.}}
{\qed\end{trivlist}}

\newcommand{\cu}[1]{{\textcolor{red}{#1}}}
\newcommand{\tout}[1]{{}}
\newcommand{\good}{{\rm good}}
\newcommand{\bad}{{\rm bad}}

\newcommand{\id}{\openone}

\title{Well-conditioned multiproduct {H}amiltonian simulation}
\author{Guang Hao Low}
\affiliation{Microsoft Quantum, Redmond WA, USA}
\author{Vadym Kliuchnikov}
\affiliation{Microsoft Quantum, Redmond WA, USA}
\author{Nathan Wiebe}
\affiliation{Microsoft Quantum, Redmond WA, USA}
\affiliation{Pacific Northwest National Laboratory, Richland WA, USA}
\affiliation{Department of Physics, University of Washington, Seattle, WA, USA}
\begin{abstract}
Product formula approximations of the time-evolution operator on quantum computers are of great interest due to their simplicity, and good scaling with system size by exploiting commutativity between Hamiltonian terms. 
However, product formulas exhibit poor scaling with the time $t$ and error $\epsilon$ of simulation as the gate cost of a single step scales exponentially with the order $m$ of accuracy.
We introduce well-conditioned multiproduct formulas, which are a linear combination of product formulas, where a single step has polynomial cost $\mathcal{O}(m^2\log{(m)})$ and succeeds  with probability $\Omega(1/\operatorname{log}^2{(m)})$.
Our multiproduct formulas imply a simple and generic simulation algorithm that simultaneously exploits commutativity in arbitrary systems and has a worst-case cost $\mathcal{O}(t\log^{2}{(t/\epsilon)})$ which is optimal up to poly-logarithmic factors. 
In contrast, prior Trotter and post-Trotter Hamiltonian simulation algorithms realize only one of these two desirable features.
A key technical result of independent interest is our solution to a conditioning problem in previous multiproduct formulas that amplified numerical errors by $e^{\Omega(m)}$ in the classical setting, and led to a vanishing success probability $e^{-\Omega(m)}$ in the quantum setting.
\end{abstract}
\maketitle






\ssection{Introduction}
Quantum computers promise to enable the efficient simulation of quantum Hamiltonian dynamics, which is, in general, intractable on classical computers. 
However, the quantum gate cost of simulating many important systems, such as quantum field theories~\cite{Jordan2012QuantumFieldTheory} and chemistry~\cite{Reiher2016Reaction}, is still prohibitive~\cite{Childs2017Speedup}.
As such, new techniques for digital Hamiltonian simulation remains a subject of intense research that has seen tremendous progress in recent years~\cite{Poulin2011timedependent,Berry2014Exponential,Berry2015Truncated,Low2016HamSim,Campbell2019Random,Low2018SpectralNorm}.
More, though, is required to realize technologically relevant applications within the foreseeable future.

State-of-art quantum algorithms for simulation may be broadly categorized as one of two complementary approaches: Lie-Trotter-Suzuki product formulas~\cite{Berry2007Efficient}, or Linear-Combination-of-Unitaries (LCU)~\cite{Low2016Qubitization}.
Both seek to approximate the unitary time-evolution operator of Schr\"{o}dinger's equation using the fewest number of primitive quantum gates.
However, they differ in asymptotic gate cost with respect to the time, error, and the size of the simulated system.
A dichotomy in this difference often makes the preferred method for any given situation unclear.

On one hand, the gate cost of LCU approaches is near-linear in time and logarithmic in error, which is essentially optimal according to no-fast-forwarding theorems~\cite{Berry2015Hamiltonian}.
In contrast, an order $2m$ product formula is more expensive by a super-logarithmic factor $(t/\epsilon)^{1/{2m}}$~\cite{Berry2007Efficient} of time and error.
As the cost of product formulas is exponential in the order $m$, such as the standard recursive construction by Suzuki~\cite{Trotter1959Product,Suzuki1990Fractal}, achieving a poly-logarithmic overhead is impossible, even by varying the order arbitrarily.

On the other hand, product formulas tend to scale significantly better with the size of typical physical systems due to the principle of locality.
A most dramatic separation in performance is observed in simulating strictly local interactions~\cite{Childs2017Speedup}.
Without any special modification, the gate cost of high-order product formulas scales almost-linearly like $\mathcal{O}(N^{1+o(1)})$ in system size~\cite{Childs2019Lattice}.
In contrast, all known LCU  approaches lose this desirable feature and exhibit quadratic scaling in general $\Omega(N^2)$~\cite{Haah2018quantum}.
Whereas the maximum stepsize of product formulas is limited by how well terms in the Hamiltonian commute, LCU approaches are unable to exploit commutation between Hamiltonian terms.
Similar advantages are observed in simulations with long-ranged interactions~\cite{Tran2019powerlaw} such as the coulomb potential~\cite{poulin2014trotter}, and systems with small algebras such as the quantum harmonic oscillator~\cite{Somma2016Trotter}.
\begin{figure}[b]
	\includegraphics[width=\columnwidth]{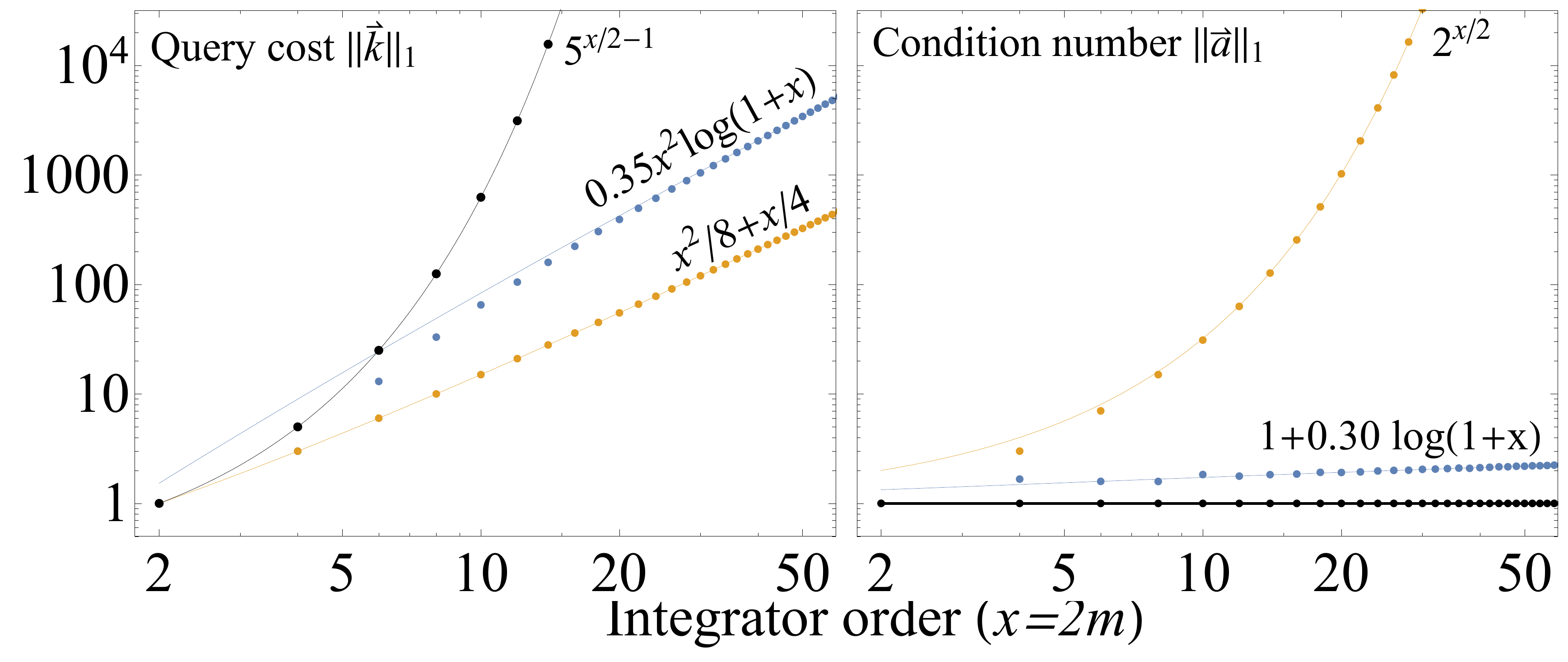}
	\caption{(left) Number of queries to a second-order product formula and (right) condition number for each step of an order $2m$ integrator. 
	(black) Trotter-Suzuki product formulas~\cite{Suzuki1990Fractal}~\cref{eq:TrotterSuzukiIntegrators}, or (yellow) multiproduct formulas by Chin~\cite{Chin2010Multiproduct} exhibit exponential scaling for at least one parameter,
	whereas (blue) our well-conditioned multiproduct formulas~\cref{eq:vandersoln,eq:roundedexponent}  combines the best properties of both.
	}
	\label{fig:algorithm}
\end{figure}

We present a simple algorithm for Hamiltonian simulation that combines the best properties of product formula and LCU approaches.
In addition to improved scaling with system size through an explicit dependence on commutators, our algorithm also matches the limits of no-fast-forwarding up to a logarithmic factor of time and error.  
The basic idea is to approximate high-order time-evolution by a linear combination of low-order product formulas, known in classical numerical techniques as a multiproduct formula~\cite{Chin2010Multiproduct}, which is a generalization of Richardson extrapolation.
Unlike the exponential cost of product formulas, we find families of order $2m$ multiproduct formula that can be realized with only $\mathcal{O}(m^2\log{(m)})$ queries to any symmetric product formula, and on a quantum computer, succeeds with a high probability $\Omega(\frac{1}{\log^2{(m)}})$ that is easily amplified using robust oblivious amplitude amplification~\cite{Berry2014Exponential}.
Thus the overhead $\mathcal{O}(m^2\log^2{(m)}(t/\epsilon)^{1/m})$ is made logarithmic by a simple optimization over the order.

Our key technical result is solving a conditioning problem in multiproduct formulas, which is of independent interest. 
This problem manifested as an exponentially small probability of success $e^{-\Omega(m)}$ in previous quantum implementations~\cite{Childs2012}, 
which arose from the exponentially precise cancellation of terms required by prior known high-order multiproduct formulas.
In contrast, the properties of our well-conditioned multiproduct formulas, illustrated in~\cref{fig:algorithm}, feature both a polynomial query cost and coefficients of logarithmic size, and are therefore numerically stable.

In the following, we outline the multiproduct conditioning problem in the context of Hamiltonian simulation.
This problem is then solved by our constructive proof that well-conditioned multiproduct formulas of arbitrary order exist, and moreover have an elegant closed-form description. 
We also provide an efficient numerical recipe for constructing optimally-conditioned instances of multiproduct formulas.
These multiproduct formulas let us prove our main claim of a Hamiltonian simulation algorithm that simultaneously has optimal scaling with respect to time and error, up to a logarithmic overhead, and also exploits commutativity of terms. 
We also validate our claims in a simple numerical benchmark depicted in~\cref{fig:errorPlot} of simulating the 1D Heisenberg model.

\ssection{The multiproduct conditioning problem} The Hamiltonian $H=\sum_{j=1}^N h_j$ of many physical systems is described by a sum of $N$ local terms. 
The first explicit quantum algorithm for approximating the unitary time-evolution operator $\overrightarrow{U}_1(\Delta)\approx e^{-iH\Delta}$ was by Lloyd~\cite{Lloyd1996universal}, and splits evolution by the whole into evolution by its parts, that is 
\begin{align}
\overrightarrow{U}_1(\Delta)=\overrightarrow{\prod}_{j=1}^N e^{-i h_j \Delta}= e^{-iH\Delta} + \mathcal{O}(\sum_{j<k}\|[h_j,h_k]\|\Delta^2)\nonumber,
\end{align}
where the ordering of terms $e^{-ih_1\Delta}e^{-ih_2\Delta}\cdots$ is indicated by the arrow `$\rightarrow$'.
As this decomposition is correct to first order, evolution for arbitrary long times $t$ is accomplished by applying $t/\Delta$ approximate segments, each comprising of $N$ exponentials, with the stepsize $\Delta$ chosen to control the overall accumulated error.
Importantly, the error term depends explicitly on pairwise commutators of Hamiltonian terms, which in turn limits the maximum stepsize.
Higher order-$\alpha$ integrators exist, such as the second-order symmetric product formula
\begin{align}
\label{eq:symmproductformula}
U_2(\Delta)=\overrightarrow{U}_1(\Delta/2)\cdot\overleftarrow{U}_1(\Delta/2)=  e^{-iH\Delta} + O(\Delta^3),
\end{align}
which is used in the recursion by Suzuki~\cite{Suzuki1990Fractal}
\begin{align}
\label{eq:TrotterSuzukiIntegrators}
U_{\alpha}(\Delta)&=U^2_{\alpha-2}(p_{\alpha}\Delta)\cdot U_{\alpha-2}((1-4p_{\alpha}\Delta)\cdot U^2_{\alpha-2}(p_{\alpha}\Delta),\nonumber\\
p_\alpha&=1/(4-4^{1/(\alpha-1)}).
\end{align}
However, these makes $5^{\alpha/2-1}$ queries to the base sequence $U_2$, and are thus impractical at high orders.

We instead focus on multiproduct formulas, where a higher order $2m$ integrator is constructed from a linear combination of any symmetric lower-order product formula
\begin{align}
\label{eq:MultiProduct}
U_{\vec{k}}(\Delta)=\sum^{M}_{j=1}a_{j} U_2^{k_j}\left(\frac{\Delta}{k_j}\right) = e^{-iH\Delta}+\mathcal{O}(\Delta^{2m+1}),
\end{align}
such as the second-order Trotter-Suzuki formula $U_2$.
Any symmetric product formula has a formal Baker-Campbell-Hausdorff (BCH) expansion $U_2(\Delta)= e^{-iH\Delta +  E_3\Delta^3 + E_5 \Delta^5+ \cdots}$ for some error operators $E_k$~\cite{Blanes2000Magnus}, which implies the Taylor expansion
\begin{align}
U^{k_j}_2(\Delta/{k_j}) &= e^{-iH\Delta} + \frac{\Delta^3}{{k_j}^2}\tilde{E}_3(\Delta) + \frac{\Delta^5}{{k_j}^4}\tilde{E}_5(\Delta) + \cdots,\nonumber
\end{align}
for some error operators $\tilde{E}_k(\Delta)$.
Thus all lower order error terms may be canceled by choices of coefficients $a_j$ that solve the $m\times M$ system of linear equations
\begin{equation}
\underbrace
{\begin{bmatrix}
	1 & 1&\cdots &1\\
	k_1^{-2} & k_2^{-2} & \cdots & k_M^{-2}\\
	\vdots & \vdots & \ddots &\vdots\\
	k_1^{-2m+2} & k_2^{-2m+2} & \cdots & k_M^{-2m+2}\\
	\end{bmatrix}}_
{V_{m,M}(\vec{k}^{-2})}
\underbrace{
	\begin{bmatrix}a_{1}\\a_{2}\\\vdots\\a_{M} \end{bmatrix}
}
_{
	\vec{a}
}=
\underbrace{
	\begin{bmatrix}1\\ 0 \\\vdots\\0 \end{bmatrix}
}
_
{\hat{e}_1}
.\label{eq:VanderMonde}
\end{equation}
The left-hand side of~\cref{eq:VanderMonde} is a Vandermonde matrix $V_{m,M}(\vec{k}^{-2})\in\mathbb{R}^{m\times M}$, where $\vec{k}^{-2}=[k_1^{-2},\cdots,k_{M}^{-2}]$.
In the square case $M=m$, this has the solution
\begin{equation}
a_j = \prod_{q=\{1,\ldots,m\}\setminus j} \frac{k_{j}^2}{k_j^2-k_q^2}=\prod_{q\neq j} \frac{1}{1-(k_q/k_j)^2}.\label{eq:vandersoln}
\end{equation}

As described by Chin~\cite{Chin2010Multiproduct}, it suffices to take the simplest choice of an arithmetic progression for the exponents $k_j=j$ with $M=m$.
By summing over the exponents, only $\|\vec{k}\|_1\in\mathcal{O}(m^2)$ queries to $U_2$ are required, which appears to be an exponential improvement over that of the Trotter-Suzuki integrators~\cref{eq:TrotterSuzukiIntegrators}.
Unfortunately, these multiproduct formulas are ill-conditioned, as reflected by the quantity we call the `condition number' $\|\vec{a}\|_1\in e^{\Omega(m)}$, which is exponentially large in the order $m$~\cite{Childs2012}.
This implies an extremely precise cancellation of terms in~\cref{eq:MultiProduct}, which amplifies any numerical error of the base sequence $U_2$ by a factor $\|\vec{a}\|_1$.

Within the quantum setting, standard linear-combination-of-unitaries techniques translate ill-conditioning into an exponentially small success probability $\|\vec{a}\|^{-2}_1$~\cite{Childs2012}.
Using the recently developed oblivious amplitude amplification technique~\cite{Berry2014Exponential}, this probability may be boosted close to unity, but still at high cost $\mathcal{O}(\|\vec{a}\|^{-1}_1)$.
This highlights the need for multiproduct formulas with small condition number.

\ssection{Solutions to the conditioning problem} 
We shed insight on the conditioning problem by considering the general under-determined setting where the exponents $k_j$ are arbitrary rather than an arithmetic sequence, and where the Vandermonde matrix is not necessarily square.
Our main technical result is an super-exponential reduction in the condition number, illustrated in~\cref{fig:algorithm} and formally stated by the following theorem.
\begin{theorem}[Well-conditioned multiproduct formulas]
	\label{Lem:ChebyshevIntegerMultiproduct}
	There exist order-$2m$ multiproduct formulas~\cref{eq:MultiProduct} with polynomial integer exponents $\|\vec{k}\|_1\in\mathcal{O}(m^2 \log{(m)})$ and logarithmic condition number $\|\vec{a}\|_1\in\mathcal{O}(\log m)$.
\end{theorem}

Our strategy for proving~\cref{Lem:ChebyshevIntegerMultiproduct} is constructive.
First, we relax~\cref{eq:MultiProduct} to allow real-valued exponents $k''_j$ with coefficients $a''_j$, and find well-conditioned solutions in closed-form for arbitrary orders $2m$, through an elegant connection to Chebyshev polynomials.
Second, we modify these solutions to obtain exponents $k'_j$ with a larger gap between consecutive exponents $|k'_j-k'_{j+1}|\in\Omega(1/m)$.
Third, we scale and round the exponents $k'_j$ to unique integers $k_j$, and show that the condition number $\|\vec{a}\|_1$ changes by at most a constant factor. 

As the $U_2$ query complexity must be at least quadratic $\ge \frac{m}{2}(1+m)$ for any choice of $m$ unique integer exponents in~\cref{eq:MultiProduct}, and largest condition number is at most unity, our result in~\cref{Lem:ChebyshevIntegerMultiproduct} is also optimal up to at most logarithmic factors.
As a bonus, we also present a rational linear program of polynomial size whose solutions describe multiproduct formulas with optimal condition number.

\begin{proof}
Consider a set of $m$ polynomials $\{p_j(x)=\sum_{i=1}^{m}A_{j,i} x^{i-1}\}_{j=1}^m$ with coefficients represented as the square matrix $A\in\mathbb{R}^{m\times m}$, that are orthogonal
\begin{align}
\label{eq:OrthogonalPolynomial}
\langle \vec{p}_j, \vec{p}_k \rangle &= \sum^m_{i=1}p_j(x_i)p_k(x_i) =\delta_{jk}\langle\vec{p}_j, \vec{p}_j \rangle,\\
\vec{p}_j&=[p_j(x_1),p_j(x_2),\cdots,p_j(x_m)],\nonumber
\end{align}
over some discrete set of interpolation points $\vec{x}$, where $\delta_{jk}$ is the Kronecker delta function. 
Now, left-multiply the Vandermonde matrix by the polynomial coefficients $A$. The $j^{\mathrm{th}}$ row of the output satisfies
\begin{align}
\label{eq:PolynomialInterpolation}
(A\cdot V_{m,m}(\vec{x})\cdot\vec{a}'')_j
=(A\cdot\hat{e}_1)_j
=A_{j,1}
=\langle \vec{p}_j,\vec{a}''\rangle.
\end{align}
Using orthogonality~\cref{eq:OrthogonalPolynomial}, the above~\cref{eq:PolynomialInterpolation} is satisfied by the choice
$\vec{a}'' = \sum_{i=1}^{m}\frac{A_{i,1} \vec{p_{i}}}{\langle p_{i}, p_{i} \rangle}$.

We find that coefficients with desirable properties are described by the basis $p_j(x)\equiv T_{j-1}(2x-1)$ of Chebyshev polynomials $T_j(x)=\cos{(j \cos^{-1}(x))}$ which are orthogonal
$\langle p_{j}, p_{k} \rangle = \frac{m}{2}\delta_{jk}(1+\delta_{j1})$ with respect to the Chebyshev interpolation points 
\begin{align}
x_j^{(m)}=\sin^{2}\left(\frac{\pi(2j-1)}{4m}\right)=1/k_j''^{2}.
\end{align}
By substitution, the coefficients are given by 
\begin{align}
\label{Eq:ChebyshevMultiProductCoefficients}
a_j''^{(m)}= \frac{(-1)^{j+1}}{m}\cot{\left(\frac{\pi(2j-1)}{4m}\right)},\;\;j\in[m].
\end{align}
We drop the superscript ${(m)}$ indicating the multiproduct formula order $2m$ whenever the context is clear.
Thus we may bound the exponents $\|\vec{k}''\|_1\in\Theta(m\log{(m)})$ and condition number $\|\vec a''\|_1\in\Theta(\log{(m)})$.

An intermediate real-exponent solution that is important to obtaining our rounded integer-exponent solution drops the latter half of the Chebyshev interpolation points.
Choose $x_j'^{(m)}=1/k'^{2}_j=x_j^{(2m)}$, where $j\in[m]$.
Then from~\cref{eq:vandersoln},
	\begin{align}
	| a_j'^{(m)}|
	\le 
	\left|\prod_{q=m+1}^{2m} \frac{a_{j}'^{(m)}}{1-x_j^{(2m)}/x_{q}^{(2m)}}\right|
	=
	|a_j''^{(2m)}|,\;j \in [m],\nonumber
	\end{align}
	which follows from the monotonicity of ${x}_j$. 
Thus the exponents and condition number are also bounded by
	$\|\vec{a}'\|_1\in\mathcal{O}(\log{(m)})$
	and $\|\vec{k}'\|_1\in\Theta(m\log{(m)})$.

\begin{figure*}
	\includegraphics[width=0.666\columnwidth]{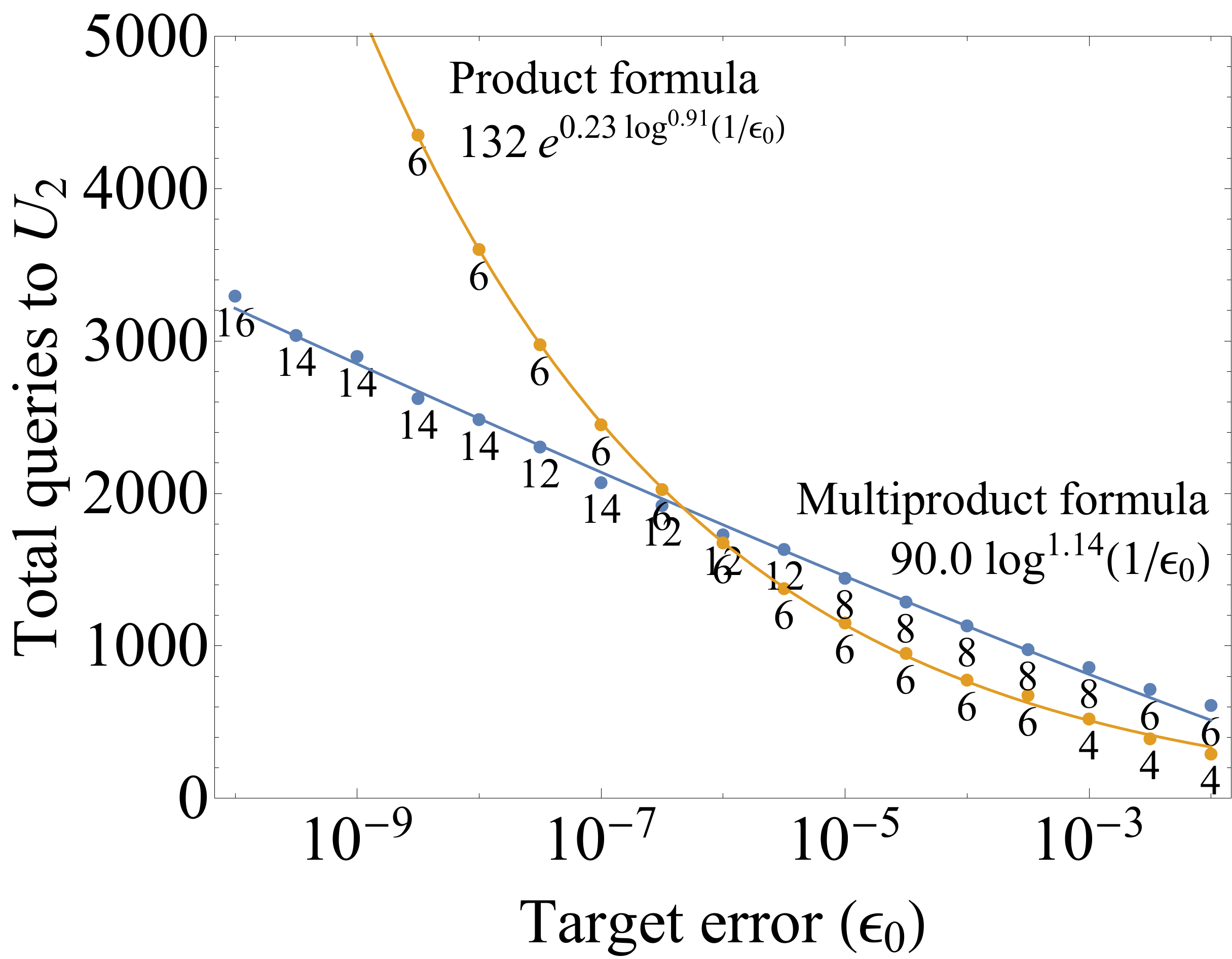}
	\includegraphics[width=1.33\columnwidth]{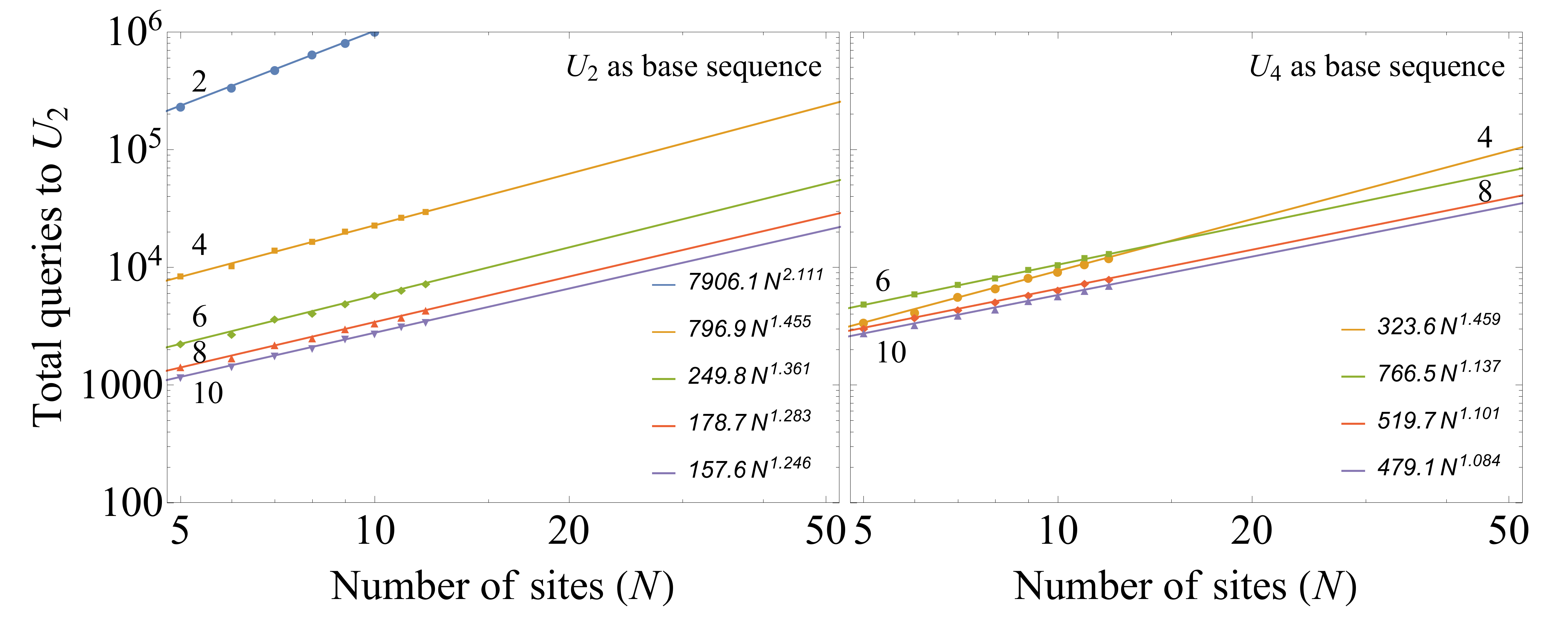}
	\caption{ Total $U_2$ query cost, including the cost of oblivious amplitude amplification, to simulate the Heisenberg chain. (Left) Simulation for time $t=N$ on $N=10$ sites as a function of error using Trotter-Suzuki product formulas (yellow), or multiproduct formulas a base sequence $U_2$ (blue). Each point is labeled by the order of the integrator that is used. 
		(Middle and right)
		Simulation for time $t=N$ with error $\epsilon=10^{-8}$ as a function of system size using multiproduct formulas with either $U_2$ or $U_4$ as the base sequence. Each line is labeled by the order of the applied multiproduct formula.
	}
	\label{fig:errorPlot}
\end{figure*}

Implementing fractional $U_2$ queries, though asymptotically efficient in principle~\cite{Gilyen2018singular}, can be impractical. 
Thus, we choose for an order $2m$ multiproduct formula the rounded exponents 
\begin{align}
\label{eq:roundedexponent}
k_j = \lceil K k'_j\rceil=\left\lceil K/\sqrt{x_j^{(2m)}}\right\rceil,\;\;j\in[m],
\end{align}
where the scale factor $K<\sqrt{8}m/\pi$ ensures rounding to unique integers, and implies $\|\vec{k}\|_1\in\mathcal{O}(m^2\log{(m)})$.

We now prove that the coefficients $a_j$ change by at most a multiplicative constant compared to $a'_j$.
As $K k'_j\in\Theta(m^2/j)\subseteq\Omega(m)$, the fractional shift $|k_j-K k'_j|/m\in\Theta(1/m)$ is small for large orders.
Thus the fractional change in $\gamma_q=(k'_q/k'_j)^2$ from Taylor's theorem is
\begin{align}
\left(\frac{k_{q}}{k_{j}}\right)^2=\left(\frac{k'^{2}_{q}}{k'^{2}_{j}}\right)^2\left(1+\Delta[q,j]\right)
,\;
|\Delta[q,j]|\in\Theta\left(\frac{|q-j|}{m^2}\right),\nonumber
\end{align}
where the sign of $\Delta[q,j]$ matches the sign of $(q-j)$.
As $\Delta[q,j]$ is also small, the shift in $a'_j$ to leading order is given by the derivatives $\frac{\partial a'_j}{\partial\gamma_q}=\frac{a'_j}{1-\gamma_q}$, following~\cref{eq:vandersoln}. Thus
\begin{align}
\label{eq:coefficient_shift}
\frac{|a'_j-a_j|}{|a_j|}&\in\Theta\left(\sum_{q\neq j}\frac{\Delta[q,j]}{1-\gamma_q}\right)
\subseteq\Theta\left(\sum_{q\neq j}\frac{x'_q|q-j|}{m^2|x'_q-x'_j|}\right)\nonumber
\\
&\subseteq\Theta(1).
\end{align}
We complete our proof by evaluating $\|\vec{k}\|_1$ using~\cref{eq:roundedexponent}, and noting that~\cref{eq:coefficient_shift} implies $\|\vec{a}\|_1\in \Theta(\|\vec{a}'\|_1)$.
\end{proof}

We can further optimize the condition number $\|\vec a\|_1$ by a numerical search over to all ${M}\choose{m}$ subsets of exponents $\{k_j\}_{j=1}^m\subseteq[M]$. 
This can be cast as an efficient linear program
\begin{align}
\min_{\vec{a}} \|\vec a\|_1\;\text{s.t.}\;{V_{m,M}(\vec{k}^{-2})}\cdot \vec{a}=\hat{e}_1\;\wedge\;k_j=j\in[M],
\label{eq:OptimizationProblem}
\end{align}
followed by minimizing with respect to $m\in[M]$. 
Analogous to sparse signal recovery~\cite{Rudelson2006Sparse}, one-norm minimization ensures that the solution $\vec{a}$ is sparse with exactly $m$ non-zero elements.
We tabulate the solutions to~\cref{eq:OptimizationProblem} in~\cref{sec:Tabulated_Results}.

\ssection{Hamiltonian simulation in the worst-case}
In the worst-case where all terms are maximally non-commutative, well-conditioned multiproduct formulas translate into Hamiltonian simulation algorithms that match no-fast-forwarding up to a logarithmic factors in time and error, stated formally in the following. 
\begin{theorem}[Hamiltonian simulation by well-conditioned multiproduct formulas]
	\label{Thm:HamSimByMPF}
	Time evolution can be approximated with error $\|U^r_{\vec{k}}(t/r) - e^{-itH}\|\le\epsilon\le1$, where $r=\Theta(t\lambda)$ and $\lambda=\sum_{j=1}^{N}\|h_j\|$, by a quantum circuit that succeeds with probability $1-\mathcal{O}(\epsilon)$ using
	$\mathcal{O}\left(t\lambda \log^{2}{(t\lambda/\epsilon)}\right)$
	controlled-$U_2$ queries
	and $\mathcal{O}(t\lambda \log{(t\lambda/\epsilon)})$ additional quantum gates.
\end{theorem}
\begin{proof}
The basic idea is to bound the error of a single multiproduct step, followed by varying the multiproduct order sub-logarithmically with time and error. 
We begin with the remainder of a product formula~\cite{Berry2007Efficient,Childs2017Speedup}
\begin{align}
\label{eq:remainder_error}
\left\|\mathcal{R}_{2m}\big[U_2^j(\Delta/j)-e^{-iH\Delta}\big]\right\|\le \frac{2|\Delta\lambda|^{2m+1}}{(2m+1)!}e^{|\Delta\lambda|}.
\end{align}
%
By a triangle inequality, the error
\begin{align}
\label{eq:mpf_remainder_error}
\left\|U_{\vec{k}}(\Delta)-e^{-iH\Delta}\right\|\le \frac{2\|\vec{a}\|_1|\Delta\lambda|^{2m+1}}{(2m+1)!}e^{|\Delta\lambda|}=\epsilon_{\Delta},
\end{align}
of a single multiproduct step accumulates after $r=t/\Delta$ steps to
\begin{align}
\left\|U^{r}_{\vec{k}}(t/r)-e^{-iHt}\right\|\le \epsilon_{t/r} r(1+\epsilon_{t/r})^{r-1}\le\epsilon,
\end{align} 
which is at most $\epsilon\le 1$ with the choice 
\begin{align}
r = t\lambda \max\left\{\left(\frac{8t\lambda\|\vec{a}\|_1}{\epsilon(2m+1)!}\right)^{1/(2m)} , \frac{1}{\log{(2)}}\right\}.
\end{align}

The cost of $U^{r}_{\vec{k}}(t/r)$ is then $r\|\vec{k}\|_1$ queries to the base product formula $U_2$.
This expression is simplified using Stirling's formula $\frac{1}{(2m+1)!}\in\Theta\left({e^{e^z-ze^z-z/2}}\right)$ where $2m=e^{z}$.
We scale the order with $z=W(\log{(t\lambda/\epsilon)})$ using the Lambert-$W$ function which satisfies,  by definition, $z e^{z}=\log{(t\lambda/\epsilon)}$.
Using the well-conditioned multiproduct formulas of~\cref{Lem:ChebyshevIntegerMultiproduct} where $\|\vec{a}\|_1\in\mathcal{O}(\log{(m)})\subseteq\mathcal{O}(z)$ and $\|\vec{k}\|_1\in\Theta(m^2\log{(m)})\subseteq\Theta(ze^{2z})$,
the number of steps is $r\in\Theta(t\lambda)$.
Thus at most $r\|\vec{k}\|_1\in\Theta(t\lambda)\cdot\Theta(z e^{2z})$ queries to $U_2$ suffice to approximate $e^{-iHt}$.

The linear combination of unitaries quantum circuit implements a single multiproduct step.
This uses a coefficient state $\ket{a}=\sum_{j=1}^m\sqrt{a_j}\ket{j}$ and a unitary selector $S=\sum^m_{j=1}\ketbra{j}{j}\otimes (\mathrm{sign[a_j]})U^{k_j}(\Delta/k_j)$.
These apply the multiproduct formula $(\bra{a}\otimes I)S(\ket{a}\otimes I)=U_{\vec{k}}(\Delta)/\|\vec{a}\|_1$ with success probability $1/\|\vec{a}\|^2_1$.
The controlled-$U_2$ query cost of this step is $\|\vec{k}\|_1$ due to $S$, and the gate cost is $\mathcal{O}(m)$ due the dimension of $\ket{a}$.

A quantum circuit implements multiple multiproduct steps with high probability using robust oblivious amplitude amplification~\cite{Berry2015Truncated,Gilyen2018singular}, which boosts the success probability of each step to $1-\mathcal{O}(\epsilon/r)$ at a multiplicative cost of $\mathcal{O}(\|\vec{a}\|_1)$. 
This also increases the error by at most an absolute constant.
Thus the overall success probability of applying $r$ steps is $1-\mathcal{O}(\epsilon)$, with a total $U_2$ query cost of $r\|\vec{a}\|_1\|\vec{k}\|_1\in\mathcal{O}(t\lambda z^2 e^{2z})\subseteq\mathcal{O}(t\lambda\log^{2}(t\lambda/\epsilon))$, and a total additional gate cost of $r\|\vec{a}\|_1 m\in\mathcal{O}(t\lambda z e^{z})\subseteq\mathcal{O}(t\lambda\log(t\lambda/\epsilon))$ that is sub-dominant.
\end{proof}
Note that the query cost may be reduced by an additional factor of $z\in\mathcal{O}(\log\log{(t/\lambda/\epsilon)})$ using more specialized query models outlined in~\cref{sec:circuit_optimization}.

\ssection{Cost dependence on Hamiltonian term commutators} As multiproduct formulas use product formulas as their base sequence, their costs also inherit an explicit scaling with the commutators of Hamiltonian terms, as captured by the following.
\begin{theorem}
	\label{Thm:MPF_Commutator_error}(Commutator dependence of multiproduct formulas)
	The error $U_{\vec{k}}(\Delta)e^{i\Delta H} - I$ of an order $2m$ multiproduct formula $U_{\vec{k}}$ with $U_\alpha$ as the base sequence depends on Hamiltonian terms that all occur in commutators $[h_{j_\beta},[\cdots,[h_{j_2},h_{j_1}]\cdots]]$ nested to depth $\beta>\alpha$.
\end{theorem}
\begin{proof}
The error of a multiproduct formula that is correct to order $2m$ is the sum of remainders of product formulas.
Thus it suffices to examine the commutator structure of $W(T)\equiv U^k_{\alpha}(\Delta/k)e^{i\Delta H}$, which, for some piece-wise constant time-dependent Hamiltonian, solves the time-dependent Schr\"{o}dinger equation
$i\frac{\partial}{\partial s}W(s)=\Delta A(s)W(s),
W(0)=I
$ 
at time $T=2kN+1$.
For instance, $U_2$ is generated by
\begin{align}
\label{eq:pf_Hamiltonian}
A_k(s)&=
\begin{cases}
-H,& s\in[0,1), \\
h_{\lfloor s\rfloor}/(2k),& s\in[1,2kN+1],\\
0,&\text{otherwise},
\end{cases}\\
h_s&\equiv h_{1+( s-1\mod{(2N)})}.\nonumber
\end{align}

Following a decomposition theorem by Lam~\cite{Lam988TimeOrdered}, any time-ordered exponential has the expansion
$W(T)=\sum_{n=0}^{\infty}\Delta^n \sum_{l=1}^{n}\sum_{m_1,\cdots,m_l>0,\|\vec{m}\|_1=n}\xi(\vec{m})C_{m_1}C_{m_2}\cdots C_{m_l}$, where 
$
\xi(\vec{m})=\prod_{i=1}^{\operatorname{dim}{\vec{m}}}\left[\sum_{j=i}^{\operatorname{dim}{\vec{m}}} m_j\right]^{-1}$ and
\begin{align}
\label{eq:Lam_Decomposition}
C_{m}
=\int_{T\ge s_m\ge\cdots\ge s_1\ge 0}\mathrm{d}\vec{s}\; 
[A(s_m),[\cdots,[A(s_2),A(s_1)]\cdots]].\nonumber
\end{align}
The advantage of this expression is the explicit dependence on nested commutators $C_{m}$ of Hamiltonian terms.
Any product formula that is correct to order $\alpha$ satisfies $U_\alpha(\Delta)e^{i\Delta H}-I\in\mathcal{O}(\Delta^{\alpha+1})$ by definition. 
This implies that $C_1=C_2=\cdots=C_\alpha=0$.
Thus $W(T)$ explicitly depends only on higher-ordered nested commutators. 
\end{proof}

\ssection{Heisenberg model benchmark}
In practice, the optimal step-size $\Delta$ is determined empirically, such as by extrapolation from smaller to larger instances~\cite{poulin2014trotter, Childs2017Speedup}.
Using the 1D Heisenberg chain $H=\sum_{j=1}^{N}(X_j X_{j+1}+Y_j Y_{j+1}+Z_j Z_{j+1})$ with periodic boundary conditions,
we numerically validate in~\cref{fig:errorPlot} the logarithmic scaling of cost with maximum allowable error $\|U^r_{\vec{k}}(t/r)-e^{-iHt}\|\le\epsilon$ and that cost with respect to system size no worse than product formula used as the base sequence.

For each maximum error threshold $\epsilon$, we minimize the cost $3r\|\vec{k}\|$, where the factor three is from oblivious amplitude amplification, over all multi-product formulas tabulated in~\cref{sec:Tabulated_Results} that are optimized for oblivious amplitude amplification.
For each choice of multiproduct formula, we apply binary search to find the maximum number of steps $r$ within the allowable error $\epsilon$.

\ssection{Conclusion}
We have constructed well-conditioned multiproduct formulas that simultaneously exploit the commutativity structure of the simulated Hamiltonian and achieve a logarithmic cost dependence on error in the worst-case.
Variations of our approach are possible and worth investigating in future work.
For instance, rigorous error bounds on the nested commutators would be of practical relevance, which also have a non-trivial dependence on the order of product formulas chosen as the base sequence.
An extension to the time-dependent case is also possible by bootstrapping off existing product formulas for time-dependent simulation~\cite{Wiebe2010}.
More broadly, continued research in capitalizing on the features of average-case Hamiltonians~\cite{Low2017USA,Low2018IntPicSim,Childs2019Lattice,Tran2019powerlaw} will be crucial to the practical realization of quantum simulation on quantum computers.

\bibliographystyle{apsrev4-1}
\nocite{apsrev41Control}
\bibliography{apsrev-control,mpf}
\onecolumngrid
\appendix

\section{Examples of optimal multi-product formulas}
\label{sec:Tabulated_Results}
In this section, we tabulate coefficients for multiproduct formulas.
We provide optimized solutions where the base sequence is either a symmetric second order product formula in~\cref{Tabl:LPOAASolutionsTwo}, or a symmetric fourth order product formula in~\cref{Tabl:LPOAASolutionsFour}.
Symmetric order-$\alpha$ product formulas may be formally expressed as
$U_\alpha(\Delta)= e^{-iH\Delta +  E_{\alpha+1}\Delta^{\alpha+1} + E_{\alpha+3} \Delta^{\alpha+3}+ \cdots}$ for some error matrices $E_k$~\cite{Blanes2000Magnus}. This implies the Taylor expansion
\begin{align}
U^{k_j}_\alpha(\Delta/{k_j}) &= e^{-iH\Delta}
+ \frac{\Delta^{\alpha+1}}{{k_j}^\alpha}\tilde{E}_{\alpha+1}(\Delta) + \frac{\Delta^{\alpha+3}}{{k_j}^{\alpha+2}}\tilde{E}_{\alpha+3}(\Delta) + \cdots.\nonumber
\end{align}
Thus all error terms of order $2m$ and below cancel in the linear combination
\begin{align}
U_{\vec{k}}(\Delta)=\sum^{M}_{j=1}a_{{k_j}} U_\alpha^{k_{k_j}}\left(\frac{\Delta}{k_{k_j}}\right) = e^{-iH\Delta}+\mathcal{O}(\Delta^{2m+1}),
\end{align}
if the coefficients $a_j$ and exponents $k_j$ satisfy the following under-determined system of linear equations.
\begin{equation}
\begin{bmatrix}
1 & 1&\cdots &1\\
k_1^{-\alpha} & k_2^{-\alpha} & \cdots & k_M^{-\alpha}\\
k_1^{-\alpha-2} & k_2^{-\alpha-2} & \cdots & k_M^{-\alpha-2}\\
\vdots & \vdots & \ddots &\vdots\\
k_1^{-2m+2} & k_2^{-2m+2} & \cdots & k_M^{-2m+2}\\
\end{bmatrix}
\begin{bmatrix}a_{1}\\a_{2}\\a_{3}\\\vdots\\a_{M} \end{bmatrix}
\begin{bmatrix}1\\ 0 \\ 0 \\ \vdots\\0 \end{bmatrix}
.\label{eq:genVanderMonde}
\end{equation}

These coefficients are optimized through the linear program of~\cref{eq:OptimizationProblem} to minimize the number of queries to the base sequence, which is captured by the product $\|\vec{a}\|_1\|\vec{k}\|_1$, where $\vec{k}$ only contains the $k_j$ exponents that correspond to non-zero $a_j$ coefficients in~\cref{eq:MultiProduct}. 
Use of oblivious amplitude amplification~\cite{Berry2015Truncated,Gilyen2018singular} requires rounding $\|\vec{a}\|_1$ up to the smallest value $n$ that satisfies 
\begin{align}
n=\operatorname{argmin}_{1\le y\le \|\vec{a}\|_1}\left\lceil\frac{\pi}{4 \sin^{-1}(1/y)}-\frac{1}{2}\right\rceil\in\mathbb{Z}_{\mathrm{odd}},
\end{align}
For all cases that we tabulate in~\cref{sec:Tabulated_Results}, it turns out that $\|\vec{a}\|_1$ is small enough that $n$ is always $3$. 
Thus we also provide solutions that minimize $\|\vec{k}\|_1$ for the largest $\|\vec{a}\|_1\le2$.

\begin{table*}[t]
	\begin{tabularx}{\linewidth}{c|c|c|>{\raggedright\arraybackslash}X}
		$m$ & $\|\vec{a}\|_1$ & $\|\vec{k}\|_1$ & Non-zero coefficients of optimized multi-product formulas $U_{\vec{k}}(\Delta)=\sum^{M}_{j=1}a_{j} U_2^{k_j}\left(\frac{\Delta}{k_j}\right) = e^{-iH\Delta}+\mathcal{O}(\Delta^{2m+1})$ \\
		\hline\hline
		2 & 1.667 & 3 & $\vec{k}=(1,2)$, $\vec{a}=(${\tiny$\frac{-1}{3}$, $\frac{4}{3}$}$)$ \\
		\hline
		3 & 1.333 & 9 & $\vec{k}=(1,2,6)$, $\vec{a}=(${\tiny$\frac{1}{105}$, $\frac{-1}{6}$, $\frac{81}{70}$}$)$ \\
		\hline
		4 & 1.401 & 16 & $\vec{k}=(1,2,3,10)$, $\vec{a}=(${\tiny$\frac{-1}{2376}$, $\frac{2}{45}$, $\frac{-729}{3640}$, $\frac{31250}{27027}$}$)$ \\
		\hline
		5 & 1.373 & 28 & $\vec{k}=(1,2,3,5,17)$, $\vec{a}=(${\tiny$\frac{1}{165888}$, $\frac{-256}{89775}$, $\frac{6561}{179200}$, $\frac{-390625}{2128896}$, $\frac{6975757441}{6067353600}$}$)$ \\
		\hline
		6 & 1.530 & 37 & $\vec{k}=(1,2,3,4,6,21)$, $\vec{a}=(${\tiny$\frac{-1}{5544000}$, $\frac{8}{19665}$, $\frac{-81}{4480}$, $\frac{65536}{669375}$, $\frac{-216}{875}$, $\frac{7626831723}{6537520000}$}$)$ \\
		\hline
		7 & 1.365 & 58 & $\vec{k}=(1,2,3,4,5,9,34)$, $\vec{a}=(${\tiny$\frac{1}{798336000}$, $\frac{-8}{654885}$, $\frac{59049}{41108480}$, $\frac{-1048576}{52518375}$, $\frac{244140625}{4596673536}$, $\frac{-31381059609}{192832640000}$, $\frac{4660977897838088}{4131462743533125}$}$)$ \\
		\hline
		8 & 1.372 & 78 & $\vec{k}=(1,2,3,4,5,6,12,45)$, $\vec{a}=(${\tiny$\frac{-1}{87524236800}$, $\frac{32}{66844575}$, $\frac{-729}{5017600}$, $\frac{131072}{28477575}$, $\frac{-48828125}{1520031744}$, $\frac{23328}{425425}$, $\frac{-95551488}{622396775}$, $\frac{1532278301220703125}{1360389650333249536}$}$)$ \\
		\hline
		9 & 1.357 & 102 & $\vec{k}=(1,2,3,4,5,6,8,15,58)$, $\vec{a}=(${\tiny$\frac{1}{14351497574400}$, $\frac{-4}{328930875}$, $\frac{59049}{6613376000}$, $\frac{-4194304}{7439025825}$, $\frac{6103515625}{831680898048}$, $\frac{-59049}{2452450}$, $\frac{274877906944}{5654031508125}$, $\frac{-360406494140625}{2342511781722112}$, $\frac{250246473680347348787521}{222930340909804639361250}$}$)$ \\
		\hline
		10 & 1.359 & 128 & $\vec{k}=(1,2,3,4,5,6,7,10,18,72)$, $\vec{a}=(${\tiny$\frac{-1}{2405702668723200}$, $\frac{1}{3304192500}$, $\frac{-177147}{328182400000}$, $\frac{16777216}{244314672525}$, $\frac{-152587890625}{88665552847872}$, $\frac{177147}{14314300}$, $\frac{-1628413597910449}{64065702729600000}$, $\frac{152587890625}{3090381882588}$, $\frac{-7625597484987}{50102940387500}$, $\frac{33537732413930512368795648}{30010892586441560158990625}$}$)$ \\
		\hline
		11 & 1.358 & 158 & $\vec{k}=(1,2,3,4,5,6,7,8,12,22,88)$, $\vec{a}=(${\tiny$\frac{1}{489053083097779200}$, $\frac{-4}{648001265625}$, $\frac{4782969}{181060880000000}$, $\frac{-524288}{79303299075}$, $\frac{95367431640625}{315052006795029504}$, $\frac{-9565938}{2362935575}$, $\frac{79792266297612001}{4243214845584000000}$, $\frac{-137438953472}{5028288890625}$, $\frac{2507653251072}{57111976796875}$, $\frac{-11119834626984462962}{75338562178830234375}$, $\frac{764149216957226040350612652032}{684739190213840837873006015625}$}$)$ \\
		\hline
		12 & 1.350 & 193 & $\vec{k}=(1,2,3,4,5,6,7,8,10,14,27,106)$, $\vec{a}=(${\tiny$\frac{-1}{144390239142589440000}$, $\frac{1}{11687273325000}$, $\frac{-129140163}{150638992824320000}$, $\frac{1073741824}{2652237835025625}$, $\frac{-95367431640625}{3025491514793459712}$, $\frac{129140163}{188838650000}$, $\frac{-79792266297612001}{15530574349255680000}$, $\frac{4503599627370496}{352223792657611875}$, $\frac{-95367431640625}{4105747222248672}$, $\frac{79792266297612001}{1838945428407945000}$, $\frac{-381520424476945831628649898809}{2599557522032734399084748800000}$, $\frac{85915027059992607611268303810877410409}{76826458670327827099700739675297300000}$}$)$ \\
		\hline
		13 & 1.376 & 224 & $\vec{k}=(1,2,3,4,5,6,7,8,9,11,16,31,121)$, $\vec{a}=(${\tiny$\frac{1}{31462070283141120000000}$, $\frac{-8192}{5071269939762151125}$, $\frac{1162261467}{30384094532599808000}$, $\frac{-68719476736}{1985545575685546875}$, $\frac{59604644775390625}{12695130964224456523776}$, $\frac{-9521245937664}{55318072575390625}$, $\frac{191581231380566414401}{84690757033018785792000}$, $\frac{-288230376151711744}{24270307981988693625}$, $\frac{984770902183611232881}{43760580453990400000000}$, $\frac{-81402749386839761113321}{3009720957024337920000000}$, $\frac{4835703278458516698824704}{112513856873765905761328125}$, $\frac{-620412660965527688188300451573157121}{4173437971930764175429651660800000000}$, $\frac{801795320536133573571931534665380233173841533961}{715886828276024991553383459008280526848000000000}$}$)$ \\
		\hline
		14 & 1.343 & 271 & $\vec{k}=(1,2,3,4,5,6,7,8,9,10,13,19,37,147)$, $\vec{a}=(${\tiny$\frac{-1}{12947955743587345367040000}$, $\frac{4096}{255968570807777953125}$, $\frac{-387420489}{439666742394880000000}$, $\frac{4398046511104}{2954011765908438860625}$, $\frac{-59604644775390625}{176199152307255474388992}$, $\frac{793437161472}{40323113771565625}$, $\frac{-191581231380566414401}{473539419700199424000000}$, $\frac{295147905179352825856}{85657194918952954171875}$, $\frac{-8862938119652501095929}{703740150628892016640000}$, $\frac{122070312500000000000}{7299214348532926488669}$, $\frac{-91733330193268616658399616009}{4911660455320856706416640000000}$, $\frac{1768453418076865701195582595329481}{47461833812963668327615011225600000}$, $\frac{-59325966985223687799599734398071581327609}{424325455586783326123540037039554560000000}$, $\frac{696762206271866268428168706860580089445450671440761}{625438255717644276228669818265124841896017920000000}$}$)$ \\
		\hline
		15 & 1.340 & 316 & $\vec{k}=(1,2,3,4,5,6,7,8,9,10,12,15,22,42,170)$, $\vec{a}=(${\tiny$\frac{1}{5708934616140416641204224000}$, $\frac{-1}{6792655122878625000}$, $\frac{4782969}{256012756679680000000}$, $\frac{-8589934592}{147809328285209716125}$, $\frac{476837158203125}{22035963278560349454336}$, $\frac{-4782969}{2477563088000}$, $\frac{1341068619663964900807}{22848100216623931392000000}$, $\frac{-18446744073709551616}{25299400178786092734375}$, $\frac{8862938119652501095929}{2254956509146578485248000}$, $\frac{-476837158203125}{58708966584812544}$, $\frac{4565043429507072}{332784274650953125}$, $\frac{-253410816192626953125}{11861507876217629966336}$, $\frac{144209936106499234037676064081}{3575238879657630656278080000000}$, $\frac{-712698848302837170849772887}{5102848721743078683520000000}$, $\frac{13518854368623590892663053210288544178009033203125}{12158976319252260175611999331938261282704299352064}$}$)$ \\
		\hline\hline
		3 & 1.889 & 7 & $\vec{k}=(1,2,4)$, $\vec{a}=(${\tiny$\frac{1}{45}$, $\frac{-4}{9}$, $\frac{64}{45}$}$)$ \\
		\hline
		4 & 1.913 & 13 & $\vec{k}=(1,2,3,7)$, $\vec{a}=(${\tiny$\frac{-1}{1152}$, $\frac{64}{675}$, $\frac{-729}{1600}$, $\frac{117649}{86400}$}$)$ \\
		\hline
		5 & 1.826 & 23 & $\vec{k}=(1,2,3,5,12)$, $\vec{a}=(${\tiny$\frac{1}{82368}$, $\frac{-64}{11025}$, $\frac{243}{3200}$, $\frac{-390625}{959616}$, $\frac{3981312}{2977975}$}$)$ \\
		\hline
		6 & 1.972 & 32 & $\vec{k}=(1,2,3,4,6,16)$, $\vec{a}=(${\tiny$\frac{-1}{3213000}$, $\frac{2}{2835}$, $\frac{-2187}{69160}$, $\frac{4096}{23625}$, $\frac{-4374}{9625}$, $\frac{4294967296}{3273645375}$}$)$ \\
		\hline
		7 & 1.966 & 46 & $\vec{k}=(1,2,3,4,5,9,22)$, $\vec{a}=(${\tiny$\frac{1}{333849600}$, $\frac{-32}{1091475}$, $\frac{59049}{17024000}$, $\frac{-1048576}{21560175}$, $\frac{244140625}{1865493504}$, $\frac{-31381059609}{72289817600}$, $\frac{100429708055072}{74479301134875}$}$)$ \\
		\hline
		8 & 1.979 & 61 & $\vec{k}=(1,2,3,4,5,6,11,29)$, $\vec{a}=(${\tiny$\frac{-1}{30481920000}$, $\frac{128}{92542905}$, $\frac{-177147}{417464320}$, $\frac{16777216}{1227909375}$, $\frac{-6103515625}{62538448896}$, $\frac{22674816}{131718125}$, $\frac{-379749833583241}{970056622080000}$, $\frac{297558232675799463481}{228244622609203200000}$}$)$ \\
		\hline
		9 & 1.961 & 80 & $\vec{k}=(1,2,3,4,5,6,8,14,37)$, $\vec{a}=(${\tiny$\frac{1}{5082098112000}$, $\frac{-2}{58046625}$, $\frac{1594323}{62664448000}$, $\frac{-4194304}{2589134625}$, $\frac{152587890625}{7155594141696}$, $\frac{-6377292}{89810875}$, $\frac{274877906944}{1833170464125}$, $\frac{-66465861139202}{162925121341125}$, $\frac{12337511914217166362274241}{9423091227263641095168000}$}$)$ \\
		\hline
		10 & 1.960 & 102 & $\vec{k}=(1,2,3,4,5,6,7,10,18,46)$, $\vec{a}=(${\tiny$\frac{-1}{981682644096000}$, $\frac{1}{1347192000}$, $\frac{-1594323}{1202574464000}$, $\frac{67108864}{397105891875}$, $\frac{-152587890625}{35937133360128}$, $\frac{1594323}{52052000}$, $\frac{-1628413597910449}{25788472744320000}$, $\frac{152587890625}{1225454342112}$, $\frac{-1853020188851841}{4489223458720000}$, $\frac{3244150909895248285300369}{2448694306609618747680000}$}$)$ \\
		\hline
		11 & 1.965 & 126 & $\vec{k}=(1,2,3,4,5,6,7,8,12,22,56)$, $\vec{a}=(${\tiny$\frac{1}{198008706639744000}$, $\frac{-4}{262214465625}$, $\frac{4782969}{73196816000000}$, $\frac{-524288}{32016859875}$, $\frac{95367431640625}{126975876815563776}$, $\frac{-9565938}{950324375}$, $\frac{232630513987207}{4962824380800000}$, $\frac{-68719476736}{1005657778125}$, $\frac{2507653251072}{22484083496875}$, $\frac{-1345499989865120018402}{3329964448304296359375}$, $\frac{15986247194027594038116352}{12080430539009984869306875}$}$)$ \\
		\hline
		12 & 1.991 & 152 & $\vec{k}=(1,2,3,4,5,6,7,8,10,14,26,66)$, $\vec{a}=(${\tiny$\frac{-1}{51894975564432000000}$, $\frac{1}{4197360384000}$, $\frac{-43046721}{18010899530624000}$, $\frac{268435456}{237412593928125}$, $\frac{-95367431640625}{1080806878190527488}$, $\frac{43046721}{22422400000}$, $\frac{-79792266297612001}{5513246098918272000}$, $\frac{1125899906842624}{31132733992234875}$, $\frac{-95367431640625}{1436932748579328}$, $\frac{79792266297612001}{624032382916800000}$, $\frac{-3211838877954855105157369}{7744161043187653392000000}$, $\frac{3504121441488212268351778470441}{2635228428724696794547667200000}$}$)$ \\
		\hline
		13 & 1.982 & 180 & $\vec{k}=(1,2,3,4,5,6,7,8,9,11,16,30,78)$, $\vec{a}=(${\tiny$\frac{1}{12242003387236761600000}$, $\frac{-1}{240754472083125}$, $\frac{59049}{599844044800000}$, $\frac{-4294967296}{48164829960121875}$, $\frac{476837158203125}{39357888848075685888}$, $\frac{-12754584}{28663009375}$, $\frac{191581231380566414401}{32683848606724423680000}$, $\frac{-18014398509481984}{583844880799805625}$, $\frac{12157665459056928801}{207303317917696000000}$, $\frac{-9849732675807611094711841}{138697287360986537164800000}$, $\frac{302231454903657293676544}{2602510298491324924828125}$, $\frac{-84470272064208984375}{217380860126701862591}$, $\frac{64103685379672408817299868943378}{48924902981004808782903378259375}$}$)$ \\
		\hline
		14 & 1.978 & 213 & $\vec{k}=(1,2,3,4,5,6,7,8,9,10,13,19,35,91)$, $\vec{a}=(${\tiny$\frac{-1}{4439277486085452595200000}$, $\frac{4096}{87717923176556413125}$, $\frac{-10460353203}{4064796621689323520000}$, $\frac{4398046511104}{1010347097562618440625}$, $\frac{-2384185791015625}{2407041263119343026176}$, $\frac{21422803359744}{371150155673171875}$, $\frac{-558545864083284007}{469625870777057280000}$, $\frac{295147905179352825856}{29063796548915485231875}$, $\frac{-79766443076872509863361}{2142768079487107072000000}$, $\frac{976562500000000000}{19689884915517097281}$, $\frac{-542800770374370512771595361}{9676704180632135600701440000}$, $\frac{1768453418076865701195582595329481}{15163688298677075616718651392000000}$, $\frac{-1864347957889074348926544189453125}{4726418762579442113057948568649728}$, $\frac{303179125313825004671598231974400041446691527}{230988684060471966545405562424513265664000000}$}$)$ \\
		\hline
		15 & 1.996 & 248 & $\vec{k}=(1,2,3,4,5,6,7,8,9,10,12,15,22,40,104)$, $\vec{a}=(${\tiny$\frac{1}{1937737434439902268293120000}$, $\frac{-8}{18438181244486128125}$, $\frac{129140163}{2344018344571904000000}$, $\frac{-536870912}{3130127505257878125}$, $\frac{2384185791015625}{37292604904429619576832}$, $\frac{-258280326}{45224496191875}$, $\frac{459986536544739960976801}{2645000331324645158092800000}$, $\frac{-1125899906842624}{520219729019615625}$, $\frac{79766443076872509863361}{6823033117248825917440000}$, $\frac{-95367431640625}{3938334372013041}$, $\frac{23110532361879552}{561942053651103125}$, $\frac{-4105255222320556640625}{63409238245639889354752}$, $\frac{144209936106499234037676064081}{1133393565815662450671009421875}$, $\frac{-53687091200000000000000000000}{131927979694795165375304717907}$, $\frac{8727375299152525849208662144232535364890263552}{6622476769886717511303522653639563434013753125}$}$)$ \\
		\hline
	\end{tabularx}
	\caption{\label{Tabl:LPOAASolutionsTwo}Multi-product solutions to~\cref{eq:MultiProduct} using a symmetric second-order product formula as the base sequence, where $\vec{k}$ only contains the $k_j$ exponents that correspond to non-zero $a_j$ coefficients, that (top half) minimize $\|\vec{a}\|_1\|\vec{k}\|_1$, and (bottom half) minimize $\|\vec{k}\|_1$ such that $\|\vec a\|_1\le 2$.}
\end{table*}

\begin{table*}[t]
	\begin{tabularx}{\linewidth}{c|c|c|>{\raggedright\arraybackslash}X}
		$m$ & $\|\vec{a}\|_1$ & $\|\vec{k}\|_1$ & Non-zero coefficients of optimized multi-product formulas $U_{\vec{k}}(\Delta)=\sum^{M}_{j=1}a_{j} U_4^{k_j}\left(\frac{\Delta}{k_j}\right) = e^{-iH\Delta}+\mathcal{O}(\Delta^{2m+1})$ \\
		\hline\hline
		3 & 1.133 & 3 & $\vec{k}=(1,2)$, $\vec{a}=\big(${\tiny$\frac{-1}{15}$, $\frac{16}{15}$}$\big)$ \\
		\hline
		4 & 1.169 & 7 & $\vec{k}=(1,2,4)$, $\vec{a}=\big(${\tiny$\frac{1}{945}$, $\frac{-16}{189}$, $\frac{1024}{945}$}$\big)$ \\
		\hline
		5 & 1.130 & 13 & $\vec{k}=(1,2,3,7)$, $\vec{a}=\big(${\tiny$\frac{-1}{72576}$, $\frac{256}{42525}$, $\frac{-729}{11200}$, $\frac{823543}{777600}$}$\big)$ \\
		\hline
		6 & 1.153 & 20 & $\vec{k}=(1,2,3,4,10)$, $\vec{a}=\big(${\tiny$\frac{1}{4633200}$, $\frac{-4}{8775}$, $\frac{59049}{3312400}$, $\frac{-32768}{429975}$, $\frac{7812500}{7378371}$}$\big)$ \\
		\hline
		7 & 1.162 & 29 & $\vec{k}=(1,2,3,4,5,14)$, $\vec{a}=\big(${\tiny$\frac{-1}{422884800}$, $\frac{16}{711585}$, $\frac{-531441}{210277760}$, $\frac{1048576}{32021325}$, $\frac{-244140625}{3115034496}$, $\frac{221460595216}{211290425775}$}$\big)$ \\
		\hline
		8 & 1.146 & 41 & $\vec{k}=(1,2,3,4,5,7,19)$, $\vec{a}=\big(${\tiny$\frac{1}{69424128000}$, $\frac{-4096}{7059362625}$, $\frac{1594323}{9777152000}$, $\frac{-67108864}{15008560875}$, $\frac{1220703125}{54428516352}$, $\frac{-678223072849}{9927650304000}$, $\frac{799006685782884121}{760815408697344000}$}$\big)$ \\
		\hline
		9 & 1.132 & 55 & $\vec{k}=(1,2,3,4,5,6,9,25)$, $\vec{a}=\big(${\tiny$\frac{-1}{12031358976000}$, $\frac{512}{36014090805}$, $\frac{-177147}{17595719680}$, $\frac{268435456}{447210547875}$, $\frac{-6103515625}{855151976448}$, $\frac{90699264}{4518292625}$, $\frac{-22876792454961}{388865404928000}$, $\frac{931322574615478515625}{890959478088566734848}$}$\big)$ \\
		\hline
		10 & 1.136 & 70 & $\vec{k}=(1,2,3,4,5,6,7,11,31)$, $\vec{a}=\big(${\tiny$\frac{1}{2043368570880000}$, $\frac{-1024}{2909265333975}$, $\frac{14348907}{23348779417600}$, $\frac{-2147483648}{28359549343125}$, $\frac{3814697265625}{2103852280578048}$, $\frac{-14693280768}{1202197871875}$, $\frac{1628413597910449}{71380702376755200}$, $\frac{-5559917313492231481}{99574372143267840000}$, $\frac{699053619999045038539170241}{669907499211984479846400000}$}$\big)$ \\
		\hline
		11 & 1.124 & 89 & $\vec{k}=(1,2,3,4,5,6,7,9,13,39)$, $\vec{a}=\big(${\tiny$\frac{-1}{566665233039360000}$, $\frac{8192}{1566265058733375}$, $\frac{-177147}{8182300672000}$, $\frac{68719476736}{13379472211730625}$, $\frac{-95367431640625}{441712566705586176}$, $\frac{1451188224}{583039331875}$, $\frac{-1628413597910449}{196103758675968000}$, $\frac{5559060566555523}{257415179141120000}$, $\frac{-8650415919381337933}{161394976296271872000}$, $\frac{1532395228870645870817151}{1476256552793711575040000}$}$\big)$ \\
		\hline
		12 & 1.129 & 108 & $\vec{k}=(1,2,3,4,5,6,7,8,10,16,46)$, $\vec{a}=\big(${\tiny$\frac{1}{130657826528570880000}$, $\frac{-1}{10646268979500}$, $\frac{387420489}{416359454038528000}$, $\frac{-67108864}{155258089284375}$, $\frac{95367431640625}{2897576688177819648}$, $\frac{-387420489}{558615557500}$, $\frac{3909821048582988049}{779185106411272704000}$, $\frac{-70368744177664}{5916663985357125}$, $\frac{95367431640625}{5138570770721388}$, $\frac{-295147905179352825856}{5697428575450811574375}$, $\frac{907846434775996175406740561329}{872274131863076152425452737500}$}$\big)$ \\
		\hline
		13 & 1.124 & 131 & $\vec{k}=(1,2,3,4,5,6,7,8,9,12,19,55)$, $\vec{a}=\big(${\tiny$\frac{-1}{43447863763854950400000}$, $\frac{8192}{7073068018875271875}$, $\frac{-43046721}{1596721705779200000}$, $\frac{8589934592}{361249657415521875}$, $\frac{-95367431640625}{30638789854214750208}$, $\frac{352638738432}{3245980816203125}$, $\frac{-27368747340080916343}{20527267517536665600000}$, $\frac{1152921504606846976}{182086598692188815625}$, $\frac{-109418989131512359209}{10740621348071014400000}$, $\frac{369768517790072832}{23616211610169303125}$, $\frac{-4898762930960846817716295277921}{97524264824097994835656704000000}$, $\frac{939343707638512715789016819000244140625}{903512885077539117948496998119881310208}$}$\big)$ \\
		\hline
		14 & 1.119 & 156 & $\vec{k}=(1,2,3,4,5,6,7,8,9,10,14,22,65)$, $\vec{a}=\big(${\tiny$\frac{1}{15242285039483761459200000}$, $\frac{-1}{74028867154790625}$, $\frac{10460353203}{14214237399025664000000}$, $\frac{-137438953472}{112277799175720303125}$, $\frac{476837158203125}{1751356624892462235648}$, $\frac{-10460353203}{679322825431250}$, $\frac{191581231380566414401}{627798629880211046400000}$, $\frac{-9223372036854775808}{3725527047551068134375}$, $\frac{79766443076872509863361}{9418308093415584563200000}$, $\frac{-95367431640625}{9332044267838292}$, $\frac{191581231380566414401}{13787827585583083621875}$, $\frac{-1191817653772720942460132761}{25447746446924438066105812500}$, $\frac{8748372096373426118698083496952056884765625}{8437250966399640675632741317018303623856128}$}$\big)$ \\
		\hline
		15 & 1.123 & 182 & $\vec{k}=(1,2,3,4,5,6,7,8,9,10,12,16,25,74)$, $\vec{a}=\big(${\tiny$\frac{-1}{6212990716825347366912000000}$, $\frac{16}{118790275159464418125}$, $\frac{-3486784401}{205492433943689052160000}$, $\frac{4294967296}{82237356069233990625}$, $\frac{-59604644775390625}{3108738756535924985561088}$, $\frac{13947137604}{8303784373259375}$, $\frac{-459986536544739960976801}{9218005522325675570626560000}$, $\frac{2305843009213693952}{3829570149292051539375}$, $\frac{-717897987691852588770249}{229502829147119976448000000}$, $\frac{59604644775390625}{9672474243773408067}$, $\frac{-14975624970497949696}{1683122287482811800625}$, $\frac{618970019642690137449562112}{37197162160941289276160015625}$, $\frac{-2220446049250313080847263336181640625}{44975293463351494115118116335062614016}$, $\frac{81217248802771228597652036390959994837496721}{78240999048919289336803622198030389875515625}$}$\big)$ \\
		\hline\hline
		4 & 1.610 & 6 & $\vec{k}=(1,2,3)$, $\vec{a}=\big(${\tiny$\frac{1}{336}$, $\frac{-32}{105}$, $\frac{729}{560}$}$\big)$ \\
		\hline
		5 & 1.526 & 11 & $\vec{k}=(1,2,3,5)$, $\vec{a}=\big(${\tiny$\frac{-1}{22464}$, $\frac{256}{12285}$, $\frac{-2187}{8320}$, $\frac{390625}{314496}$}$\big)$ \\
		\hline
		6 & 1.642 & 17 & $\vec{k}=(1,2,3,4,7)$, $\vec{a}=\big(${\tiny$\frac{1}{1365120}$, $\frac{-256}{159975}$, $\frac{59049}{884800}$, $\frac{-262144}{821205}$, $\frac{282475249}{225244800}$}$\big)$ \\
		\hline
		7 & 1.908 & 24 & $\vec{k}=(1,2,3,4,5,9)$, $\vec{a}=\big(${\tiny$\frac{-1}{94003200}$, $\frac{128}{1237005}$, $\frac{-59049}{4874240}$, $\frac{524288}{3132675}$, $\frac{-244140625}{552738816}$, $\frac{31381059609}{24395571200}$}$\big)$ \\
		\hline
		8 & 1.887 & 34 & $\vec{k}=(1,2,3,4,5,7,12)$, $\vec{a}=\big(${\tiny$\frac{1}{14707630080}$, $\frac{-128}{46139625}$, $\frac{177147}{222208000}$, $\frac{-65536}{2900205}$, $\frac{6103515625}{51404709888}$, $\frac{-678223072849}{1612182528000}$, $\frac{11609505792}{8770136375}$}$\big)$ \\
		\hline
		9 & 1.974 & 44 & $\vec{k}=(1,2,3,4,5,6,8,15)$, $\vec{a}=\big(${\tiny$\frac{-1}{1621638144000}$, $\frac{32}{297604125}$, $\frac{-59049}{749056000}$, $\frac{4194304}{844333875}$, $\frac{-1220703125}{18930143232}$, $\frac{1889568}{8960875}$, $\frac{-274877906944}{651070294875}$, $\frac{72081298828125}{56715799748608}$}$\big)$ \\
		\hline
		10 & 1.905 & 58 & $\vec{k}=(1,2,3,4,5,6,7,11,19)$, $\vec{a}=\big(${\tiny$\frac{1}{390029230080000}$, $\frac{-1024}{552406511475}$, $\frac{14348907}{4394293657600}$, $\frac{-2147483648}{5269941455625}$, $\frac{3814697265625}{384412809166848}$, $\frac{-14693280768}{214999159375}$, $\frac{1628413597910449}{12429674378035200}$, $\frac{-5559917313492231481}{14481005254410240000}$, $\frac{104127350297911241532841}{79380043802925465600000}$}$\big)$ \\
		\hline
		11 & 1.946 & 73 & $\vec{k}=(1,2,3,4,5,6,7,9,13,23)$, $\vec{a}=\big(${\tiny$\frac{-1}{94661139234816000}$, $\frac{8192}{260671653309375}$, $\frac{-14348907}{109614202880000}$, $\frac{68719476736}{2193181992626625}$, $\frac{-95367431640625}{71563821261520896}$, $\frac{117546246144}{7539771113875}$, $\frac{-79792266297612001}{1506850787819520000}$, $\frac{150094635296999121}{1039844174200832000}$, $\frac{-19004963774880799438801}{45404688739675668480000}$, $\frac{1716155831334586342923895201}{1307159711825499154022400000}$}$\big)$ \\
		\hline
		12 & 1.985 & 88 & $\vec{k}=(1,2,3,4,5,6,7,8,10,15,27)$, $\vec{a}=\big(${\tiny$\frac{1}{18572005645575782400}$, $\frac{-128}{192858045778125}$, $\frac{4782969}{722108907520000}$, $\frac{-8589934592}{2762256589663725}$, $\frac{3814697265625}{15882764789219328}$, $\frac{-306110016}{59389755425}$, $\frac{3909821048582988049}{102458013930455040000}$, $\frac{-36028797018963968}{386999438625174375}$, $\frac{48828125000000}{309209171790519}$, $\frac{-3649115753173828125}{9252770793176694784}$, $\frac{42391158275216203514294433201}{32701084707427617351270400000}$}$\big)$ \\
		\hline
		13 & 1.974 & 107 & $\vec{k}=(1,2,3,4,5,6,7,8,9,12,18,32)$, $\vec{a}=\big(${\tiny$\frac{-1}{6142641675702884352000}$, $\frac{32}{3894861002484375}$, $\frac{-4782969}{24887553152000000}$, $\frac{134217728}{786303703560375}$, $\frac{-59604644775390625}{2643142970256498597888}$, $\frac{306110016}{384438179125}$, $\frac{-191581231380566414401}{19327329750643200000000}$, $\frac{4503599627370496}{94033072774265625}$, $\frac{-450283905890997363}{5723189781557248000}$, $\frac{641959232274432}{4838745927640625}$, $\frac{-14409084988511915616}{36169625809113671875}$, $\frac{1267650600228229401496703205376}{970913591804648156179749609375}$}$\big)$ \\
		\hline
		14 & 1.917 & 128 & $\vec{k}=(1,2,3,4,5,6,7,8,9,10,14,21,38)$, $\vec{a}=\big(${\tiny$\frac{1}{2211237782923205836800000}$, $\frac{-1}{10718298389019375}$, $\frac{43046721}{8441043283312640000}$, $\frac{-137438953472}{16125888405767765625}$, $\frac{59604644775390625}{31248973295537455890432}$, $\frac{-14348907}{131994362500}$, $\frac{3909821048582988049}{1797927650111324160000}$, $\frac{-9223372036854775808}{517009533455642101875}$, $\frac{984770902183611232881}{15925711282492211200000}$, $\frac{-59604644775390625}{786670162680239838}$, $\frac{3909821048582988049}{34564324896752990625}$, $\frac{-56101658612759777297212443}{153735289078666600448000000}$, $\frac{1768453418076865701195582595329481}{1379974635261707514524145253102500}$}$\big)$ \\
		\hline
		15 & 1.972 & 150 & $\vec{k}=(1,2,3,4,5,6,7,8,9,10,12,16,24,43)$, $\vec{a}=\big(${\tiny$\frac{-1}{900823728009896570880000000}$, $\frac{128}{137582627676102984375}$, $\frac{-14348907}{122123339706368000000}$, $\frac{1073741824}{2958609985221515625}$, $\frac{-298023223876953125}{2226595086494836377255936}$, $\frac{2754990144}{233628817296875}$, $\frac{-459986536544739960976801}{1303994601146388602880000000}$, $\frac{18014398509481984}{4198223988849515625}$, $\frac{-26588814358957503287787}{1181569860211220480000000}$, $\frac{19073486328125000000}{425602110643296646059}$, $\frac{-46221064723759104}{695475103338640625}$, $\frac{38685626227668133590597632}{286811608006618857584296875}$, $\frac{-258486928873495606591488}{651874549770874351796875}$, $\frac{5459046029871041534743397034507955551606171601}{4193316609494557038531498092653135994880000000}$}$\big)$ \\
		\hline
	\end{tabularx}
	\caption{\label{Tabl:LPOAASolutionsFour}Multi-product solutions to~\cref{eq:MultiProduct} using a symmetric fourth-order product formula as the base sequence,, where $\vec{k}$ only contains the $k_j$ exponents that correspond to non-zero $a_j$ coefficients, that (top half) minimize $\|\vec{a}\|_1\|\vec{k}\|_1$, and (bottom half) minimize $\|\vec{k}\|_1$ such that $\|\vec a\|_1\le 2$.}
\end{table*}

\section{Multiproduct circuit optimizations}
\label{sec:circuit_optimization}
Given a multiproduct formula $U_{\vec{k}}(\Delta)=\sum^{M}_{j=1}a_{j} U_2^{k_j}\left(\frac{\Delta}{k_j}\right)$ from~\cref{eq:MultiProduct},
we expressed cost in terms of a sum of queries to the base product formula.  
This contributes a multiplicative factor of $\|\vec{k}\|_1=k_1+k_2+\cdots+k_M$. 
However, it should be recognized that this factor is a worst-case bound.
In this section, we briefly discuss a query model where this multiplicative factor is reduced from  $\|\vec{k}\|_1$ to $\max_j k_j$.
Though this corresponds to improving the query complexity of~\cref{Thm:HamSimByMPF} by only a sub-logarithmic factor, this improvement could nevertheless be situationally advantageous.

For a Hamiltonian $H=h_1+h_2\cdots+h_N$ with $N$ terms, let us define the programmable rotation
\begin{align}
\operatorname{Prog}_k(\Delta)=\sum_{p=0}^P\ketbra{p}{p}\otimes e^{-ip\Delta h_k/P}.
\end{align}
Controlled on a binary number $p\in[0,P]$, this applies time evolution by $\Delta$ scaled by a fraction $p/P$.
Using $N$ programmable rotations in sequence, this allows us to synthesize the programmable product formula
\begin{align}
\operatorname{ProgPF}(\Delta)=\sum_{p=0}^P\ketbra{p}{p}\otimes U_2(p\Delta/ P).
\end{align}
We may then synthesize
\begin{align}
\operatorname{ProgMPF}(\Delta)=\sum_{j=1}^M\ketbra{j}{j}\otimes\sum_{p=0}^P\ketbra{p}{p}\otimes U^{k_j}_2(p\Delta/ P).
\end{align}
Note that this uses $\max_j k_j$ queries to the programmable product formula rather than $\|\vec{k}\|_1$. 
As the exponents $k_j$ are known beforehand, we may create a circuit where controlled on $\ket{j}$, we simply do not apply the remaining $M-k_j$ programmable product formulas.

Thus a single-step of the multiproduct formula 
\begin{align}
(\bra{G}\otimes I)\cdot \operatorname{ProgMPF}(\Delta)\cdot (\ket{G'}\otimes I)=\frac{1}{\|\vec{a}\|_1}\sum^{M}_{j=1}a_{j} U_2^{k_j}\left(\frac{\Delta}{k_j}\right),\nonumber
\end{align}
is approximated to error $\epsilon$ by using the states 
\begin{align}
\ket{G}&=\sum_{j=1}^M\sqrt{a_j}\ket{j}\ket{P/k_j},\\\nonumber
\ket{G'}&=\sum_{j=1}^M\sqrt{-a_j}\ket{j}\ket{P/k_j},
\end{align}
and choosing $P\in\mathcal{O}(NM\Delta/\epsilon)$.

\end{document}